\documentclass[a4paper, 14pt]{extarticle}
\usepackage[margin=0.7in]{geometry} %Delete this line in the final version.

\usepackage{amssymb,amsmath,amsfonts,amsthm,graphicx,verbatim,enumerate}
\usepackage{hyperref}

\theoremstyle{plain}
\newtheorem{proposition}{Proposition}[section]
\newtheorem{theorem}[proposition]{Theorem}

\newtheorem{lemma}[proposition]{Lemma}
\newtheorem{conjecture}[proposition]{Conjecture}
\newtheorem{problem}[proposition]{Problem}
\def\mc{\mathcal}
\def\mb{\mathbf}
\DeclareMathOperator{\gn}{gn}
\DeclareMathOperator{\Win}{Win}
\DeclareMathOperator{\Image}{Im}
\DeclareMathOperator{\Reverse}{Reverse}
\DeclareMathOperator{\Sh}{Sh}
\DeclareMathOperator{\ZY}{ZY}
\DeclareMathOperator{\DFZ}{DFZ}
\DeclareMathOperator{\Ingl}{Ingl}

\begin{document}

\title{Graph Guessing Games and non-Shannon Information Inequalities }

\author{Rahil Baber
\thanks{School of Electronic Engineering and Computer Science,
Queen Mary, University of London, London, E1 4NS, U.K. Email: 
rahilbaber@hotmail.com.}
%rahilbaber <rahilbaber@hotmail.com>
%Rahil Baber <rahilbaber@outlook.com>
%
\and Demetres Christofides
\thanks{School of Sciences,
UCLan Cyprus, 7080 Pyla, Cyprus. Email:
d.christofides@uclan.ac.uk}
\and Anh N. Dang
\thanks{School of Electronic Engineering and Computer Science,
Queen Mary, University of London, London, E1 4NS, U.K. Email:
anh.dang@eecs.qmul.ac.uk}
\and S{\o}ren Riis
\thanks{School of Electronic Engineering and Computer Science,
Queen Mary, University of London, London, E1 4NS, U.K. Email:
s.riis@qmul.ac.uk}
\and Emil R. Vaughan
\thanks{School of Electronic Engineering and Computer Science,
Queen Mary, University of London, London, E1 4NS, U.K. Email:
e.vaughan@qmul.ac.uk} }

\date\today

\maketitle

\begin{abstract}
Guessing games for directed graphs were introduced by Riis \cite{Riis07} for studying multiple unicast network coding problems. In a guessing game, the players toss generalised dice and can see some of the other outcomes depending on the structure of an underlying digraph. They later guess simultaneously the outcome of their own die. Their objective is to find a strategy which maximises the probability that they all guess correctly. The performance of the optimal strategy for a  graph is measured by the guessing number of the digraph.

%In a guessing game, we have an underlying digraph $G$ and a positive integer $s$. Players sitting on the vertices of $G$ toss $s$-sided dice. Each player can see the outcomes of all dice tosses in its in-neighbourhood and nothing else. Without any further communication, the players simultaneously try to guess the outcome of their own coin toss. The objective of the game is for the players to agree on a strategy beforehands in order to maximise the probability that all of them guess correctly.

%As it turns out that there are strategies which can significantly improve on the strategy in which every players guesses a number at random. The performance of the optimal strategy for a  graph is measured by the guessing number. (There is one such guessing number for each value of $s$ but these numbers converge to a limit as $s$ tends to infinity.)

In \cite{Christofides&Markstrom11}, Christofides and Markstr\"om studied guessing numbers of undirected graphs and defined a strategy which they conjectured to be optimal. One of the main results of this paper is a disproof of this conjecture. 

The main tool so far for computing guessing numbers of graphs is information theoretic inequalities. The other main result of the paper is that Shannon's information inequalities, which work particularly well for a wide range of graph classes, are not sufficient for computing the guessing number.

%The above two results make the concept of guessing numbers a lot more interesting, as even for undirected graphs, determining the guessing number seems to be much more difficult than what was initially thought.

Finally we pose a few more interesting questions some of which we can answer and some which we leave as open problems.

\end{abstract}

\section{Introduction}

Consider the following 2-player cooperative game: Two players toss a coin with each player seeing the outcome of the coin toss of the other player (but not their own). Then, they simultaneously guess the outcome of their own coin toss. The players win the game if they both guess correctly. Of course, if they both guess at random, then the probability of winning is $1/4$. It turns out that the players can use the extra information they have in order to improve the probability of success. For example, if they agree beforehand to follow the strategy `guess what you see' then the probability of success increases to $1/2$. We can generalise this game (see Section \ref{Sec:Definitions}) to guessing games with multiple players in which each player sees the outcome of the coin tosses (or more generally of many-sided dice throws) of other players, according to an underlying digraph. 

These guessing games \cite{Riis07,Riis07A} emerged from studying network coding problems \cite{AhlswedeCLY00} where the network is multiple unicast, i.e.~where each sender has precisely one corresponding receiver who wishes to obtain the sender's message, and a constrain that only one message can be sent through each channel at a time. A multiple unicast can be represented by a directed acyclic graph with $n$ inputs/outputs and $m$ intermediate nodes. By merging the vertices which represent the senders with their corresponding receiver vertices %in the directed acyclic graph representing the multiple unicast network,
we can create an auxiliary directed graph which has the nice property that there is no longer any distinction between router, sender, or receiver vertices. Due to the way guessing games are defined, coding functions on the original network can be translated into strategies for the guessing game on the auxiliary graph and vice versa. The performance of the optimal strategy for a guessing game is measured by the guessing number which we will define precisely in Section \ref{Sec:Definitions}.

One of the first applications of guessing games was the disproval in \cite{Riis07} of two conjectures raised by Valiant \cite{Valiant} in circuit complexity in which he asked about the optimal Boolean circuit for a Boolean function. %For more information about network coding and guessing games see \cite{Yeung08}, \cite{DFZ07}, \cite{Riis07}, and \cite{Christofides&Markstrom11}.

In this paper we provide a counterexample to a conjecture of Christofides and Markstr\"om given in \cite{Christofides&Markstrom11} which states that the optimal strategy for the guessing game of an undirected graph is based on the fractional clique cover number of the graph. (See Section \ref{Sec:Definitions} for more details.) Additionally, we will show that the guessing number for undirected graphs cannot be determined by considering only the {S}hannon information inequalities as explained in Section \ref{Sec:Entropy}. We will also make and investigate the Superman conjecture which suggests that the (asymptotic) guessing number of an undirected graph does not increase when a directed edge is added. Finally we will provide a possible example of a directed graph whose guessing number changes when its edges are reversed.

The outline of our paper is as follows. In Section \ref{Sec:Definitions} we introduce the formal language of guessing games. Section \ref{Sec:Asymptotic} is concerned with the asymptotic behaviour of guessing numbers. In Section \ref{Sec:FractionalCliqueCover} we formally define the fractional clique cover strategy from \cite{Christofides&Markstrom11} which provides a feasible computational method for calculating lower bounds of guessing numbers for undirected graphs. In Section \ref{Sec:Entropy} we introduce a method for calculating upper bounds of guessing numbers by making use of entropic arguments. Our main results appear in Section \ref{Sec:Main}. We then discuss some of the technical details of the computer searches we carried out in Section \ref{Sec:Speed}. We conclude with some open problems in Section \ref{Sec:Problems}.

\section{Definitions}\label{Sec:Definitions}

A \emph{directed graph}, or \emph{digraph} for short, is a pair
$G=(V(G),E(G))$, where $V(G)$ is the set of \emph{vertices} of $G$ and $E(G)$ is a
set of ordered pairs of vertices of $G$ called the \emph{directed edges} of $G$. Given a directed edge $e = (u,v)$, which we also denote by $uv$, we call $u$ the
\emph{tail} and $v$ the \emph{head} of $e$ and say that $e$ goes from $u$ to $v$.

For the purposes of guessing games we will assume throughout that our digraphs are loopless, i.e.\ they contain no edges of the form $uu$ for $u \in V(G)$. Once we define the guessing game it will be easily seen that the probability of winning on a digraph $G$ is equal to the probability of winning on the subgraph of $G$ obtained by removing all vertices with loops.

%Once we define the guessing game it will be easily seen that the guessing number of a directed graph $G$ is equal to the guessing number of the subgraph $H$ of $G$ obtained from $G$ by removing all vertices of $G$ which have loops.

Given a digraph $G$ and a vertex $v\in V(G)$, the
\emph{in-neighbourhood} of $v$ is $\Gamma^-(v) = \{u : uv \in
E(G)\}$ and the \emph{out-neighbourhood}
of $v$ is $\Gamma^+(v) = \{u : vu\in E(G)\}$.

In this paper our main results will primarily be on \emph{undirected
graphs} which are naturally treated as a special type of digraph $G$ where
$uv\in E(G)$ if and only if $vu\in E(G)$. We call the pair of
directed edges $uv$ and $vu$, the \emph{undirected edge} $uv$. A major role in our guessing strategies will be played by \emph{cliques} i.e.\ subgraphs in which every pair of
vertices are joined by an undirected edge.

Given a digraph $G$ and an integer $t\geq 1$, the \emph{$t$-uniform
blowup of $G$} which we will write as $G(t)$ is a digraph formed by
replacing each vertex $v$ in $G$ with a class of $t$ vertices
$v_1,\ldots,v_t$ with $u_iv_j\in E(G(t))$ if and only if $uv \in
E(G)$.

A \emph{guessing game} $(G,s)$ is a game played on a digraph $G$ and
the alphabet $A_s = \{0,1,\ldots,s-1\}$. There are $|V(G)|$ players
working as a team. Each player corresponds to one of the vertices
of the digraph. Throughout the article we will be freely speaking about the player $v$ instead of the player corresponding to the vertex $v \in V(G)$. The players know the digraph $G$, the natural number $s$, and are told to which of the vertices they correspond to. They may discuss and agree upon
a strategy using this information before the game begins, but no
communication between players is allowed after the game starts.

Once the game begins, each player $v\in V(G)$ is
assigned a value $a_v$ from $A_s$ uniformly and independently at random. The players
do not have access to their own values but can see some of the values
assigned to the other players according to the digraph $G$. To be more
precise, once the values have been assigned each player is given a
list of the players in its in-neighbourhood with their corresponding
values. Using just this information each player must guess their own
value. If all players guess correctly they will all win, but if just one
player guesses incorrectly they will all lose. The objective of the
players is to maximise their probability of winning.

As an example we consider the guessing game $(K_n,s)$, where $K_n$
is the \emph{complete (undirected) graph} of order $n$, i.e.\ $|V(K_n)|=n$ and
$E(K_n)=\{uv : u,v\in V(G), u\neq v\}$. Naively we may think that
since each player receives no information about their own value that
each player may as well guess randomly, meaning that the probability
they win is $s^{-n}$. This however is not optimal. Certainly the
probability that any given player guesses correctly is $1/s$, but Riis \cite{Riis07} noticed that by
discussing their strategies beforehand the players can in fact coordinate the moments where
they guess correctly, and therefore increase their chance of winning. For example before the game begins they can
agree that they will all play under the assumption that
\begin{align}\label{Eq:Kn_Constraint}
\sum_{v\in V(K_n)} a_v \equiv 0\bmod s.
\end{align}
Player $u$ can see all the values except its own, and assuming
(\ref{Eq:Kn_Constraint}) is true it knows that
\[
a_u \equiv -\sum_{\stackrel{v\in V(K_n)}{v\neq u}} a_v\bmod s.
\]
Consequently player $u$ will guess that its value is $-\sum_{v\in V(K_n),
v\neq u} a_v\mod s$. Hence if (\ref{Eq:Kn_Constraint}) is true every
player will guess correctly and if (\ref{Eq:Kn_Constraint}) is
false every player will guess incorrectly. So the probability they all
guess correctly is simply the probability that (\ref{Eq:Kn_Constraint}) is true
which is $1/s$. This is clearly optimal as, irrespective of the
strategy, the probability that a single player guesses correctly is
$1/s$ and so we can not hope to do better.

We note that the optimal strategy given in the example was a
\emph{pure strategy} i.e.\ there is no randomness involved in the guess
each player makes given the values it sees. The alternative is a \emph{mixed strategy} in which the players randomly choose a strategy to play from a set of pure strategies. The winning probability of the mixed strategy is the average of the winning probabilities of the pure strategies weighted according to the probabilities that they are chosen. This however is at most the maximum of the winning probabilities of the pure strategies, and so we gain no advantage by playing a mixed strategy. As such throughout this paper we will only ever consider pure strategies.

Given a guessing game $(G,s)$, for $v\in V(G)$ a \emph{strategy for
player $v$} is formally a function $f_v:A_s^{|\Gamma^-(v)|}\to A_s$
which maps the values of the in-neighbours of $v$ to an elements of $A_s$, which will be the guess of $v$. A
\emph{strategy} $\mc{F}$ for a guessing game is a sequence of such
functions $(f_v)_{v\in V(G)}$ where $f_v$ is a strategy for player
$v$. We denote by $\Win(G,s,\mc{F})$ the event that all the players guess
correctly when playing $(G,s)$ with strategy $\mc{F}$. The players'
objective is to find a strategy $\mc{F}$ that maximises
$\mb{P}[\Win(G,s,\mc{F})]$.

Rather than trying to find $\max_{\mc{F}}\mb{P}[\Win(G,s,\mc{F})]$ we
will instead work with the \emph{guessing number} $\gn(G,s)$ which
we define as
\[
\gn(G,s) = |V(G)|+\log_s
\left(\max_{\mc{F}}\mb{P}[\Win(G,s,\mc{F})]\right).
\]
Although this looks like a cumbersome property to work with we can
think of it as a measure of how much better the optimal strategy is
over the strategy of just making random guesses, as
\[
\max_{\mc{F}}\mb{P}[\Win(G,s,\mc{F})] =
\frac{s^{\gn(G,s)}}{s^{|V(G)|}}.
\]
Later we will look at information entropy inequalities as a way of
analyzing the guessing game and in this context the definition of
the guessing number will appear more natural.

\section{The asymptotic guessing number}\label{Sec:Asymptotic}

Note that the guessing number of the example $(K_n,s)$ we
discussed earlier is represented by $\gn(K_n,s) = n-1$ which does
not depend on $s$. In general $\gn(G,s)$ will depend on $s$ and it
is often extremely difficult to determine the guessing number
exactly. Consequently we will instead concentrate our efforts on
evaluating the \emph{asymptotic guessing number} $\gn(G)$ which we
define to be the limit of $\gn(G,s)$ as $s$ tends to infinity. To
prove the limit exists we first need to consider the guessing number
on the blowup of $G$.

\begin{lemma}\label{Lem:BlowupIneq}
Given a digraph $G$, and integers $s, t\geq 1$,
\[
\max_{\mc{F}}\mb{P}[\Win(G(t),s,\mc{F})] \geq
\left(\max_{\mc{F}}\mb{P}[\Win(G,s,\mc{F})]\right)^t
\]
or equivalently $\gn(G(t),s)\geq t\gn(G,s)$.
\end{lemma}

\begin{proof}
The digraph $G(t)$ can be split into $t$ vertex disjoint copies of
$G$. We can construct a strategy for $(G(t),s)$ by playing the
optimal strategy of $(G,s)$ on each of the $t$ copies of $G$ in
$G(t)$. The result follows immediately.
\end{proof}

\begin{lemma}\label{Lem:BlowupEqual}
Given a digraph $G$, and integers $s, t\geq 1$,
\[
\max_{\mc{F}}\mb{P}[\Win(G(t),s,\mc{F})] =
\max_{\mc{F}}\mb{P}[\Win(G,s^t,\mc{F})]
\]
or equivalently $\gn(G(t),s) = t\gn(G,s^t)$.
\end{lemma}

\begin{proof}
First we will show that the optimal probability of winning on
$(G,s^t)$ is at least that of $(G(t),s)$. This follows simply from
the fact that the members of the alphabet of size $s^t$, can be
represented as $t$ digit numbers in base $s$. Hence given a strategy
on $(G(t),s)$, a corresponding strategy can be played on $(G,s^t)$ by
each player pretending to be $t$ players: More precisely, if player $v$ gets assigned value $a \in A_{s^t}$, he writes it as $a_{t-1} \cdots a_1a_0$ in base $s$ and pretends to be $t$ players, say $v_0,v_1,\ldots,v_{t-1}$, where player $v_i$, for $0 \leq i \leq t-1$, gets assigned value $a_i \in A_s$. Furthermore, if player $v$ sees the outcome of player $u$, then he can construct the values assigned to the new players $u_0,u_1,\ldots,u_{t-1}$. So these new fictitious players can play the $(G(t),s)$ game using an optimal strategy. But if the fictitious players can win the $(G(t),s)$ game then the original players can win the $(G,s^t)$ game as we can reconstruct the value of $a$ from the values of $a_0,a_1,\ldots,a_{t-1}$.

A similar argument can be used to show that the optimal probability
of winning on $(G,s^t)$ is at most that of $(G(t),s)$. We will show
that for every strategy on $(G,s^t)$ there is a corresponding
strategy on $(G(t),s)$. Every vertex class of $t$ players can
simulate playing as one fictitious player by its members agreeing to
use the same strategy. The $t$ values assigned to the players in the
vertex class can be combined to give an overall value for the vertex
class. The strategy on $(G,s^t)$ can then be played allowing the
members of the vertex class to make a guess for the overall value
assigned to the vertex class. This guess will be the same for each
member as they all agreed to use the same strategy and have access
to precisely the same information. Once the guess for the vertex
class is made its value can be decomposed into $t$ values from $A_s$
which can be used as the individual guesses for each of its members.
\end{proof}

Using these results about blowups of digraphs we can show that
in some sense the guessing number is almost monotonically increasing
with respect to the size of the alphabet.

\begin{lemma}\label{Lem:Monotonic}
Given any digraph $G$, positive integer $s$, and real number
$\varepsilon>0$, there exists $t_0(G, s,\varepsilon)>0$ such that for all
integers $t\geq t_0$
\[
\gn(G,t)\geq \gn(G,s)-\varepsilon.
\]
\end{lemma}

\begin{proof}
We will prove the result by showing that
\begin{align}\label{Eq:AlmostIncLog}
\gn(G,t)\geq\frac{\lfloor \log_st\rfloor}{\log_st}\gn(G,s)
\end{align}
holds for all $t\geq s$. This will be sufficient since as $t$
increases the right hand side of (\ref{Eq:AlmostIncLog}) tends to
$\gn(G,s)$.

We will prove (\ref{Eq:AlmostIncLog}) by constructing a strategy for
$(G,t)$. Let $k=\lfloor \log_st\rfloor$ and note that $s^k$ is at
most $t$. By considering only strategies in which every player is
restricted to guess a value in $\{0,1,\ldots,s^k-1\}$ we get
\[
\max_{\mc{F}}\mb{P}[\Win(G,t,\mc{F})] \geq \mb{P}[a_v<s^k \mbox{ for
all } v\in V(G)]\max_{\mc{F}}\mb{P}[\Win(G,s^k,\mc{F})].
\]
Hence
\[
\frac{t^{\gn(G,t)}}{t^{|V(G)|}}\geq
\left(\frac{s^k}{t}\right)^{|V(G)|}
\frac{s^{k\gn(G,s^k)}}{s^{k|V(G)|}}
\]
which rearranges to
\begin{align}\label{Eq:AlmostIncreasing}
\gn(G,t)\geq\frac{k}{\log_st}\gn(G,s^k).
\end{align}
From Lemmas \ref{Lem:BlowupIneq} and \ref{Lem:BlowupEqual} we can
show $\gn(G,s^k)\geq\gn(G,s)$ which together with
(\ref{Eq:AlmostIncreasing}) completes the proof of
(\ref{Eq:AlmostIncLog}).
\end{proof}

\begin{theorem}
For any digraph $G$, $\displaystyle{\gn(G) = \lim_{s\to\infty} \gn(G,s)}$ exists.
\end{theorem}

\begin{proof}
By definition $\gn(G,s)\leq |V(G)|$ for all $s$, and $\max_{s\leq
n}\gn(G,s)$ is an increasing sequence with respect to $n$, therefore
its limit exists which we will call $\ell$. Since $\gn(G,s)\leq \ell$ for
all $s$ it will be enough to show that $\gn(G,s)$ converges to $\ell$
from below.

By the definition of $\ell$, given $\varepsilon>0$ there exists
$s_0(\varepsilon)$ such that $\gn(G,s_0(\varepsilon))\geq \ell-\varepsilon$.
From Lemma \ref{Lem:Monotonic} we know that there exists
$t_0(\varepsilon)$ such that for all $t\geq t_0(\varepsilon)$,
$\gn(G,t)\geq \gn(G,s_0(\varepsilon))-\varepsilon$ which implies
$\gn(G,t)\geq \ell-2\varepsilon$ proving we have convergence.
\end{proof}

Before we move on to the next section it is worth mentioning that
for any $s$ the guessing number $\gn(G,s)$ is a lower bound for
$\gn(G)$. This follows immediately from Lemma \ref{Lem:Monotonic}.
Furthermore for any strategy $\mc{F}$ on $(G,s)$ we have
\[
\gn(G,s) \geq |V(G)|+\log_s \mb{P}[\Win(G,s,\mc{F})].
\]
Consequently we can lower bound the asymptotic guessing number by
considering any strategy on any alphabet size.

\section{Lower bounds using the fractional clique
cover}\label{Sec:FractionalCliqueCover}

In this section we will describe a strategy specifically for
undirected graphs. As shown in the previous section this
can be used to provide a lower bound for the asymptotic guessing
number. Christofides and Markstr\"om \cite{Christofides&Markstrom11} conjectured that this bound
always equals the asymptotic guessing number.

In Section \ref{Sec:Definitions} we saw that when an undirected
graph is complete an optimal strategy is
for each player to play assuming the sum of all the values is
congruent to $0\bmod s$ (where $s$ is the alphabet size). We call this the \emph{complete graph strategy}. We can
generalise this strategy to undirected graphs which are not
complete. We simply decompose the undirected graph into vertex
disjoint cliques and then let the players play the complete graph
strategy on each of the cliques. If we are playing on an alphabet of size $s$ and we decompose the graph into $t$ cliques, then on each clique the probability of winning is $s^{-1}$ and so the probability of winning the guessing game, which is equal to the probability of winning in each of the cliques, is $s^{-t}$. Clearly the probability of
winning is higher if we choose to decompose the graph into as few
cliques as possible. The smallest number of cliques that we can decompose a graph into is called the \emph{minimum clique cover number} of $G$ and we will represent it by $\kappa(G)$. In this notation we have
\[\gn(G)\geq \gn(G,s)\geq |V(G)|-\kappa(G).\]
It is worth mentioning that finding
the minimum clique cover number of a graph is equivalent to finding
the chromatic number of the graph's complement. As such it is difficult to determine this number in the sense that the computation of the chromatic number of a graph is an NP-complete problem \cite{Karp72}.

We can improve this bound further by considering blowups of
$G$. From Lemma \ref{Lem:BlowupEqual} we know that $\gn(G,s^t) =
\gn(G(t),s)/t$, hence by the clique cover strategy on $G(t)$ we get
a lower bound of $|V(G)|-\kappa(G(t))/t$. The question is now to
determine $\min_t\kappa(G(t))/t$. We do this by looking at the fractional clique cover of $G$.

Let $K(G)$ be the set of all cliques in $G$, and let $K(G,v)$ be the
set of all cliques containing vertex $v$. A \emph{fractional
clique cover} of $G$ is a weighting $w:K(G)\to [0,1]$ such that for
all $v\in V(G)$
\[
\sum_{k\in K(G,v)} w(k) \geq 1.
\]
The minimum value of $\sum_{k\in K(G)}w(k)$ over all choices of fractional clique covers $w$ is known as \emph{the
fractional clique cover number} which we will denote by
$\kappa_f(G)$. (Although we do not define it here, we point out that the fractional clique cover number of a graph
is equal to the fractional chromatic number of its complement.)

For the purposes of guessing game strategies it will be more convenient to instead consider a special type of fractional clique cover called the regular fractional clique cover. A \emph{regular fractional clique cover} of $G$ is a weighting $w:K(G)\to [0,1]$ such that for
all $v\in V(G)$
\[
\sum_{k\in K(G,v)} w(k) = 1.
\]
The minimum value of $\sum_{k\in K(G)}w(k)$ over all choices of regular fractional clique covers $w$ can be shown to be equal to the fractional clique cover number  $\kappa_f(G)$. To see this, observe firstly that since all regular fractional clique covers are fractional clique covers the minimum value of $\sum_{k\in K(G)}w(k)$ over all choices of regular fractional clique covers $w$ is at least $\kappa_f(G)$. Finally, to show it is at most $\kappa_f(G)$ we simply observe that the optimal fractional clique cover can be made into a regular fractional cover by moving weights from larger cliques to smaller cliques. In particular, given a vertex $v$ for which $\sum_{k\in K(G,v)} w(k) > 1$ we pick a clique $k_1 \in K(G,v)$ with $w(k_1) > 0$ and proceed as follows: We change the weight of $k_1$ from $w(k_1)$ to 
\[w'(k_1) = \max\left\{0,1- \sum_{\stackrel{k\in K(G,v)}{k \neq k_1}} w(k)\right\} < w(k_1)\]
We also change the weight of the clique $k_1' = k_1 \setminus \{v\}$ from $w(k_1')$ to $w'(k_1') = w(k_1') + w(k_1) - w'(k_1)$. We leave the weight of all other vertices the same. In this way, the total sum of weights over all cliques remains the same, the total sum of weights over all cliques containing a given vertex $v' \neq v$ also remains the same, but the total sum of weights over all cliques containing $v$ is reduced. This process has to terminate because whenever we change the weight of $k_1$ it will either become equal to $0$ or the total sum of weight of all cliques containing $v$ will become equal to $1$. 

Clearly $\kappa_f(G)$ and an optimal regular fractional clique cover $w$ can be determined
by linear programming. Since all the coefficients of the constraints
and objective function are integers, $w(k)$ will be rational for all
$k\in K(G)$ as will $\kappa_f(G)$. If we let $d$ be the common
denominator of all the weights, then $dw(k)$ for $k\in K(G)$
describes a clique cover of $G(d)$. In particular it decomposes
$G(d)$ into $d\kappa_f(G)$ cliques, proving a lower bound of
\begin{equation}\label{Eq:Fractional clique cover lower bound}
\gn(G) \geq |V(G)|-\kappa_f(G).
\end{equation}
We claim that
\[\min_t \frac{\kappa(G(t))}{t} \geq\kappa_f(G)\]
and therefore we cannot hope to use regular fractional clique cover strategies to improve \eqref{Eq:Fractional clique cover lower bound}.
 To prove our claim we begin by observing that for all $t$ we have $\kappa(G(t))\geq \kappa_f(G(t))$. This is immediate as a minimal clique cover is a special type of fractional clique cover, namely one where all weights are $0$ or $1$. Hence it is enough to show that
\[\kappa_f(G(t))=t\kappa_f(G).\]
%
%
%To prove this given any optimal weighting $w$ of $K(G(t))$ consider the weight $w'$ of $G$ given by
%\[ w'(v) = \frac{1}{t}\sum_{i=1}^t w(v_i)\]
%It is easy to check that $w'$ is a  fractional clique cover of $G$ with weight $\kappa_f(G(t))/t$ and so $\kappa_f(G) \leq k_g(G(t))/t$. Conversely given an optimal weighting $w$ of $G$ the weight $\tilde{w}$ of $K(G(t))$ given by
%\[ \tilde{w}(v_i) = w(v)\]
%is a fractional clique cover of $G(t)$ with weight $t\kappa_f(G)$ showing the reverse inequality.
%
%
This can be proved simply from observing that an optimal weighting of $K(G(t))$ can always be transformed into another optimal weighting which is symmetric with respect to vertices in the same vertex class. This can be done just by moving the weights between cliques. Therefore determining $\kappa_f(G(t))$
is equivalent to determining $\kappa_f(G)$ but with the constraints
$\sum_{k\in K(G,v)}w(k)=t$ rather than $1$. The result
$\kappa_f(G(t))=t\kappa_f(G)$ is a simple consequence of this.

A useful bound on $\kappa_f(G)$ which we will make use of later is
given by the following lemma.

\begin{lemma}\label{Lem:BoundOnFractionalClique}
For any undirected graph $G$
\[
\kappa_f(G)\geq \frac{|V(G)|}{\omega(G)},
\]
where $\omega(G)$ is the number of vertices in a maximum clique in
$G$.
\end{lemma}

\begin{proof}
Let $w$ be an optimal regular fractional clique cover. Since
\[\sum_{k\in K(G,v)} w(k) = 1\] 
holds for all $v\in V(G)$, summing
both sides over $v$ gives us,
\[
\sum_{k\in K(G)} w(k)|V(k)| = |V(G)|,
\]
where $|V(k)|$ is the number of vertices in clique $k$. The result
trivially follows from observing
\[
\sum_{k\in K(G)} w(k)|V(k)| \leq \sum_{k\in K(G)} w(k)\omega(G) =
\kappa_f(G) \omega(G). \qedhere
\]
\end{proof}

The result of Christofides and Markstr\"om \cite{Christofides&Markstrom11} states the following:

\begin{theorem}\label{Thm:LowerBoundIneq}
If $G$ is an undirected graph then
\[\gn(G)\geq |V(G)|-\kappa_f(G).\]
\end{theorem}

In \cite{Christofides&Markstrom11} it was proved that the above lower bound is actually an equality for various families of undirected graphs including perfect graphs, odd cycles and complements of odd cycles. This led Christofides and Markstr\"om \cite{Christofides&Markstrom11} to conjecture that we always have equality.

\begin{conjecture} \label{Conj:LowerBoundSharp}
If $G$ is an undirected graph then
\[\gn(G) = |V(G)|-\kappa_f(G).\]
\end{conjecture}

%Christofides and Markstr\"om \cite{Christofides&Markstrom11} made the following conjecture.

%\begin{conjecture} \label{Conj:LowerBoundSharp}
%Given $G$ an undirected graph,
%\[
%\gn(G)=|V(G)|-\kappa_f(G).
%\]
%\end{conjecture}
To prove or disprove such a claim we require a way of upper bounding $\gn(G)$. This is the purpose of the next section.

\section{Upper bounds using entropy} \label{Sec:Entropy}

Recall that it is sufficient to only consider pure strategies on
guessing games. Hence given a strategy $\mc{F}$ on a guessing game
$(G,s)$ we can explicitly determine $\mc{S}(\mc{F})$ the set of all
assignment tuples $(a_v)_{v\in V(G)}$ that result in the players
winning given they are playing strategy $\mc{F}$. In this context
the players' objective is to choose a strategy that maximizes
$|\mc{S}(\mc{F})|$. We have
\begin{align*}
\gn(G,s) &= |V(G)| + \log_s\left(\max_{\mc{F}}\mb{P}[\Win(G,s,\mc{F})]\right)\\
&= |V(G)| + \max_{\mc{F}}\log_s \frac{|\mc{S}(\mc{F})|}{s^{|V(G)|}} \\
&= \max_{\mc{F}}\log_s|\mc{S}(\mc{F})|.
\end{align*}

Consider the probability space on the set of all assignment tuples
$A_s^{|V(G)|}$ with the members in $\mc{S}(\mc{F})$ occurring with
uniform probability and all other assignments occurring with $0$
probability. For each $v\in V(G)$ we define the discrete random
variable $X_v$ on this probability space to be the value assigned to
vertex $v$. The \emph{s-entropy} of a discrete random variable $X$
with outcomes $x_1,x_2,\ldots,x_n$ is defined as
\[
H_s(X) = -\sum_{i=1}^n \mb{P}[X=x_i]\log_s\mb{P}[X=x_i],
\]
where we take $0\log_s 0$ to be $0$ for consistency. Note that
traditionally entropy is defined using base $2$ logarithms, however
it will be more convenient for us to work with base $s$ logarithms. We will usually write $H(X)$ instead of $H_s(x)$. We will mention here all basic results concerning entropy that we are going to use. For more information, we refer the reader to \cite{Cover&Thomas}.

Given a set of random variables $Y_1,\ldots, Y_n$ with sets of
outcomes $\Image(Y_1),\ldots, \Image(Y_n)$ respectively, the \emph{joint
entropy} $H(Y_1,\ldots,Y_n)$ is defined as
\[
-\sum_{y_1\in\Image(Y_1)} \cdots \sum_{y_n\in\Image(Y_n)}
\mb{P}[Y_1=y_1,\ldots,Y_n=y_n]\log_s\mb{P}[Y_1=y_1,\ldots,Y_n=y_n].
\]
Given a set of random variables $Y = \{Y_1,\ldots,Y_n\}$ we will
also use the notation $H(Y)$ to represent the joint entropy
$H(Y_1,\ldots,Y_n)$. Furthermore for sets of random variables $Y$
and $Z$ we will use the notation $H(Y,Z)$ as shorthand for $H(Y\cup
Z)$. For completeness we also define $H(\emptyset)=0$.

Observe that under these definitions, the joint entropy of the set
of variables $X_G = \{X_v : v\in V(G)\}$ is
\begin{align*}
H(X_G) &= -\sum_{(a_v)\in\mc{S}(\mc{F})}
\frac{1}{|\mc{S}(\mc{F})|}\log_s\left(\frac{1}{|\mc{S}(\mc{F})|}\right)\\
&= \log_s |\mc{S}(\mc{F})|
\end{align*}
Therefore by upper bounding $H(X_G)$ for all choices of $\mc{F}$ we
can upper bound $\gn(G,s)$.

We begin by stating some inequalities that most hold regardless of
$\mc{F}$.

\begin{theorem}\label{Thm:ShannonEntropy}
Given $X,Y,Z\subset X_G$,
\begin{enumerate}
\item\label{Item:VarLowerBound} $H(X)\geq 0$.
\item\label{Item:VarUpperBound} $H(X)\leq |X|$.
\item\label{Item:VarShannon} Shannon's information inequality:
\[
H(X,Z)+H(Y,Z)-H(X,Y,Z)-H(Z)\geq 0.
\]
\item\label{Item:VarGraphConstraint} Suppose $A, B\subset V(G)$ with $\Gamma^-(u) \subset
B$ for all $u\in A$. Let $X=\{X_v : v\in A\}$ and $Y=\{X_v : v\in B\}$. Then
\[H(X,Y) = H(Y).\] 
\end{enumerate}
\end{theorem}

\begin{proof} $ $
\begin{enumerate}
\item Property \ref{Item:VarLowerBound} follows immediately from the
definition of entropy.

\item Property \ref{Item:VarUpperBound} follows from first observing that
$H(X) = \mb{E}[\log_s(1/\mb{P}[X])]$. Since the function $x \mapsto \log_s(x)$ is concave, by Jensen's inequality we get that \[H(X)\leq \log_s
\mb{E}[1/\mb{P}[X]] = \log_s |\Image(X)|\] where $\Image(X)$ is the
set of outcomes for $X$. Since $\Image(X)=A_s^{|X|}$ we have the
desired inequality $H(X)\leq |X|$.

\item Property \ref{Item:VarShannon} again follows from Jensen's
inequality. First we observe that
\begin{align*}
H(X,Z)+H(Y,Z)-H(X,Y,Z)-H(Z) =
\mb{E}_{X,Y,Z}\left[-\log_s\left(\frac{\mb{P}[X,Z]\mb{P}[Y,Z]}{\mb{P}[X,Y,Z]\mb{P}[Z]}\right)\right].
\end{align*}
By an application of Jensen's inequality this is at least
\begin{align*}
-\log_s\left(\mb{E}_{X,Y,Z}\left[\frac{\mb{P}[X,Z]\mb{P}[Y,Z]}{\mb{P}[X,Y,Z]\mb{P}[Z]}\right]\right)
&=
-\log_s\left(\sum_{X,Y,Z}\frac{\mb{P}[X,Z]\mb{P}[Y,Z]}{\mb{P}[Z]}\right)%\\
%&= -\log_s\left(\sum_{Y,Z}\mb{P}[Y,Z]\right)\\
= 0.
\end{align*}

\item Property \ref{Item:VarGraphConstraint} is a simple consequence of the fact that
the values assigned to the vertices in $A$ are completely
determined by the values assigned to the vertices in $B$. Since
$\mb{P}[X,Y]$ is either $0$ or $\mb{P}[Y]$, the result is trivially
attained by considering the definition of $H(X,Y)$ and summing over
the variables in $X$. \qedhere
\end{enumerate}
\end{proof}

From Theorem \ref{Thm:ShannonEntropy} we can form a linear program
to upper bound $H(X_G)$. In particular the linear program consists
of $2^{|V(G)|}$ variables corresponding to the values of $H(X)$ for
each $X\subset X_G$. The variables are constrained by the linear
inequalities given in Theorem \ref{Thm:ShannonEntropy} and the
objective is to maximize the value of the variable corresponding to
$H(X_G)$. We call the result of the optimization \emph{the Shannon
bound} of $G$ and denote it by $\Sh(G)$.

Note that $\Sh(G)$ can be calculated without
making any explicit use of $\mc{F}$ or $s$. Hence it is not only an
upper bound on $\gn(G,s)$ but also on $\gn(G)$.

More recently information entropy inequalities that cannot be
derived from linear combinations of Shannon's inequality (Property
\ref{Item:VarShannon} in Theorem \ref{Thm:ShannonEntropy}) have been
discovered. The first such inequality was found by Zhang and Yeung \cite{Zheng&Yeung98}.
The \emph{Zhang-Yeung inequality} states that
\begin{multline*}
-2H(A)-2H(B)-H(C)+3H(A,B)+3H(A,C)+H(A,D)+\\3H(B,C)+H(B,D)-H(C,D)-4H(A,B,C)-H(A,B,D)
\geq 0
\end{multline*}
for sets of random variables $A,B,C,D$. By setting $A=X\cup Z$,
$B=Z$, $C=Y\cup Z$, $D=Z$, the Zhang-Yeung inequality reduces to
Shannon's inequality. By replacing the Shannon inequality
constraints with those given by the Zhang-Yeung inequality we can
potentially get a better upper bound from the linear program.
However, we pay for this potentially better bound by a significant
increase in the running time of the linear program. We will call the bound
on $\gn(G)$ obtained by use of the Zhang-Yeung inequality \emph{the
Zhang-Yeung bound} and denote it by $\ZY(G)$.

In fact there are known to be infinite families of non-Shannon
inequalities even on $4$ variables. We cannot hope to add infinite
constraints to the linear program so instead we will consider the
$214$ inequalities given by Dougherty, Freiling, and
Zeger \cite[Section VIII]{DFZ}. We will refer to the resulting bound as
\emph{the Dougherty-Freiling-Zeger bound} and denote it by $\DFZ(G)$. It is perhaps worth
mentioning for those interested that the $214$
Dougherty-Freiling-Zeger inequalities imply the Zhang-Yeung
inequality (simply sum inequalities $56$ and $90$) and therefore they
also imply Shannon's inequality.

The final bound we will consider is \emph{the Ingleton bound} which we will denote by $\Ingl(G)$. This is obtained when we replace the {S}hannon inequality constraints with the \emph{Ingleton inequality}
\begin{multline*}
-H(A)-H(B)+H(A,B)+H(A,C)+H(A,D)+H(B,C)+\\H(B,D)-H(C,D)-H(A,B,C)-H(A,B,D)\geq 0.
\end{multline*}

The Ingleton inequality provides the outer-bound of the inner-cone of linearly representable entropy vectors \cite{Chan07}. By setting $A=Z,C=Y$ and $B=D=X \cup Z$, the Ingleton inequality reduces to Shannon's inequality.

If each player's strategy can be expressed as a linear combination of the values it sees, then the Ingleton inequality will hold. Therefore the inequality holds for a strategy on $(G,s^t)$ that can be represented as a linear strategy on $(G(t),s)$ (as described in the proof of Lemma \ref{Lem:BlowupEqual}). As such, the Ingleton bound gives us an
upper bound when we restrict ourselves to strategies which are
linear on the digits of the values. An important such strategy is
the fractional clique cover strategy \cite{Christofides&Markstrom11} which leads to the proof of Theorem \ref{Thm:LowerBoundIneq}.

%If each player's strategy can be expressed as a linear combination of the alues it sees then the inequality will hold. Furthermore the iequality holds for a strategy on $(G,s^t)$ that can be represented as a linear strategy on $(G(t),s)$ (as described in the proof of Lemma \ref{Lem:BlowupEqual}). As such the Ingleton bound gives us an upper bound when we restrict ourselves to strategies which are linear on the digits of the values. An important such strategy is the fractional clique cover strategy as described in Section \ref{Sec:FractionalCliqueCover}.

In searching for a counterexample to Conjecture
\ref{Conj:LowerBoundSharp} we carried out an exhaustive search on
all undirected graphs with at most $9$ vertices. We compared the
lower bound given by the fractional clique cover with the upper
bound given by the Shannon bound and in all cases the two bounds
matched. The bounds were calculated using floating point arithmetic
and so we do not claim this search to be rigorous, however it suggested
the following conjecture.

\begin{conjecture}\label{Conj:UpperBoundSharp}
If $G$ is an undirected graph then $\gn(G) = \Sh(G)$. 
\end{conjecture}

\section{Main results}\label{Sec:Main}

In this section we present our new results, most notably that both
Conjectures \ref{Conj:LowerBoundSharp} and
\ref{Conj:UpperBoundSharp} are false. Counterexamples were found by
searching through all undirected graphs on $10$ vertices or less.
For speed purposes, the search was done using floating point
arithmetic and as such there may be counterexamples that were missed
due to rounding errors. (Although this is highly unlikely, we do not
claim that it is impossible.) Despite this, we feel that it is still remarkable
that of the roughly $12$ million graphs that were checked we only
found $2$ graphs whose lower and upper bounds (given by the
fractional clique cover, and Shannon bound respectively) did not
match: the graph $R$ given in Figure \ref{Fig:RGraph}, and the graph
$R^-$ which is identical to $R$ but with the undirected edge between
vertices $9$ and $10$ removed. 

\begin{figure}[ht]
\begin{center}
\includegraphics[width=8cm]{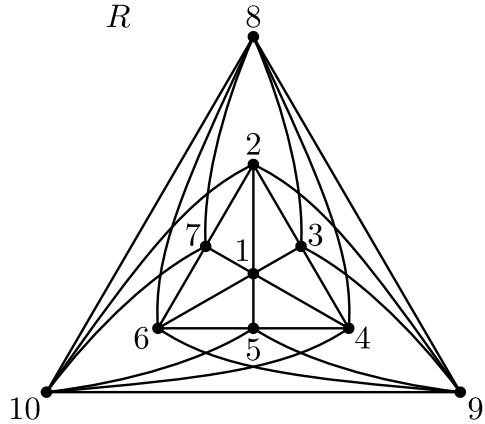}
\caption{The undirected graph $R$.}\label{Fig:RGraph}
\end{center}
\end{figure}

The graph $R$ is particularly
extraordinary as we will see that with a few simple modifications we
can create graphs which answer a few other open problems.

We begin our analysis of $R$ and $R^-$ by determining their
fractional clique cover number.

\begin{lemma}\label{Lem:RR-Lower}
We have $\kappa_f(R)=\kappa_f(R^-) = 10/3$.
\end{lemma}

\begin{proof}
By Lemma \ref{Lem:BoundOnFractionalClique} we know that
$\kappa_f(R)$ and $\kappa_f(R^-)$ are bounded below by $10/3$. To
show they can actually attain $10/3$ we need to construct explicit regular fractional clique covers whose weights add up to $10/3$.

%For $R$ we give a weight of $1/3$ to the cliques $\{1,2,3\}$,
%$\{1,4,5\}$, $\{1,6,7\}$, $\{8,9,10\}$, a weight of $2/3$ to the
%cliques $\{2,3,9\}$, $\{4,5,10\}$, $\{6,7,8\}$, and a weight of $0$
%to all other cliques.

For $R^-$ we give a weight of $1/3$ to the cliques $\{1,2,3\}$,
$\{1,4,5\}$, $\{1,6,7\}$, $\{2,3,9\}$, $\{2,7,10\}$, $\{3,8,9\}$,
$\{4,5,10\}$, $\{4,8,10\}$, $\{5,6,9\}$, $\{6,7,8\}$, and a weight
of $0$ to all other cliques. Note that this is also an optimal regular fractional clique cover for $R$.
\end{proof}

\begin{theorem}\label{Thm:R-Upper}
We have
\begin{enumerate}
\item $\Sh(R^-) = 114/17 = 6.705\ldots$
\item $\ZY(R^-) = 1212/181 = 6.696\ldots$
\item $\DFZ(R^-) = 59767/8929 = 6.693\ldots$
\item $\Ingl(R^-) = 20/3 = 6.666\ldots$
\end{enumerate}
\end{theorem}

From  Lemma \ref{Lem:RR-Lower} and Theorem \ref{Thm:R-Upper} we know
that
\[
20/3\leq \gn(R^-)\leq 59767/8929,
\]
and although we could not determine the asymptotic guessing number
exactly it does show that it does not equal the Shannon bound,
disproving Conjecture \ref{Conj:UpperBoundSharp}. Given that the
Shannon bound is not sharp we might be tempted to conjecture that
the asymptotic guessing number is the same as the Zhang-Yeung bound,
but Theorem \ref{Thm:R-Upper} also shows this to be false.
Interestingly the Ingleton bound does match the lower bound, showing
that if we restrict ourselves to only considering linear strategies
on blowups we can do no better than the fractional clique cover
strategy.

It remains an open question as to whether a non-linear
strategy on $R^-$ can do better than $20/3$ or whether by
considering the right set of entropy inequalities we can push the
upper bound down to $20/3$.

\begin{proof}[Proof of Theorem \ref{Thm:R-Upper}.]
Calculating the upper bounds involves solving rather large linear
programs. Hence the proofs are too long to reproduce here and it is
unfeasible for them to be checked by humans. Data files verifying
our claims can be provided upon request. We stress that although the
results were verified using a computer that no floating point data
types were used during the verification. Consequently no rounding
errors could occur in the calculations making the results completely
rigorous.
\end{proof}

Although $R$ is a counterexample to Conjecture
\ref{Conj:LowerBoundSharp} its optimal strategy is somewhat
complicated. So instead we will disprove the conjecture by showing a
related graph which we will call $R_{c}$ is a counterexample. The
undirected graph $R_c$ is constructed from $R$ by \emph{cloning} $3$
of its vertices. (Cloning $3$ vertices is equivalent to creating a
blowup of $R$ with $2$ vertices in $3$ of the vertex classes and
just $1$ vertex in the other classes.) The vertices we clone are $8,
9, 10$, and we label the resulting new vertices $8', 9'$, and $10'$
respectively.

\begin{theorem} \label{Thm:Rc}
We have $\gn(R_c) = 9$ while the fractional clique cover bound of $R_c$ is
$26/3 < 9$. In particular, $R_c$ provides a counterexample to Conjecture \ref{Conj:LowerBoundSharp}.
\end{theorem}

\begin{proof}
To prove that the fractional clique cover bound is $26/3$ it is enough to show that
$\kappa_f(R_c) = 13 - 26/3 = 13/3$. Lemma \ref{Lem:BoundOnFractionalClique}
tells us $\kappa_f(R_c)\geq 13/3$. It is also easy to show
$\kappa_f(R_c)\leq 13/3$ as it trivially follows from extending the regular fractional clique cover given in the proof of Lemma
\ref{Lem:RR-Lower} by giving a weight of $1$ to the clique
$\{8',9',10'\}$.

The Shannon bound of $R_c$ is $9$ proving $\gn(R_c)\leq 9$. We do
not provide the details of the Shannon bound proof as it is too long
to present here, however data files containing the proof are
available upon request.

All that remains is to prove $\gn(R_c)\geq 9$. Even though this proof was discovered partly using a computer it can be easily verified by humans. In particular, the main conclusion of this theorem, that $R_c$ is a counterexample to Conjecture \ref{Conj:LowerBoundSharp}, can be verified without the need of any computing power. 

Recall that in Section \ref{Sec:Asymptotic} we showed that the asymptotic guessing
number can be lower bounded by considering any strategy on any
alphabet size. We will take our alphabet size $s$ to be $3$. Our
strategy involves all players agreeing to play assuming the
following four conditions hold on the assigned values
\begin{align}
a_1+a_2+2a_3+a_4+2a_5+a_6+2a_7 &\equiv 0 \bmod 3,\label{Eq:Rc1}\\
a_2+a_5+a_8+a_{8'}+a_9+a_{10'} &\equiv 0 \bmod 3,\label{Eq:Rc2}\\
a_3+a_6+a_8+a_{9'}+a_{10}+a_{10'} &\equiv 0 \bmod 3,\label{Eq:Rc3}\\
a_4+a_7+a_{8'}+a_9+a_{9'}+a_{10} &\equiv 0 \bmod 3.\label{Eq:Rc4}
\end{align}
Note that the terms in (\ref{Eq:Rc1}) consist of $a_1$, and values
which player $1$ can see. Hence (\ref{Eq:Rc1}) naturally gives us a
strategy for player $1$, i.e.\ that player $1$ should guess
$-a_2-2a_3-a_4-2a_5-a_6-2a_7\mod 3$. Similarly strategies for
players $8,8',9,9',10$, and $10'$ can be achieved by rearranging
conditions (\ref{Eq:Rc3}), (\ref{Eq:Rc4}), (\ref{Eq:Rc2}),
(\ref{Eq:Rc3}), (\ref{Eq:Rc4}) and (\ref{Eq:Rc2}) respectively. A
strategy for player $2$ can be obtained by taking a linear
combination of the conditions. In particular if we sum
(\ref{Eq:Rc3}), (\ref{Eq:Rc4}), twice (\ref{Eq:Rc1}), and twice
(\ref{Eq:Rc2}) we get
\[
2a_1+a_2+2a_3+2a_7+2a_{9'}+2a_{10} \equiv 0 \bmod 3,
\]
which consists of $a_2$ and values which player $2$ can see,
allowing us to construct a strategy for player $2$. We leave it to the reader to verify that by taking the following linear combinations we obtain strategies for players $3,4,5,6$, and $7$:
\begin{itemize}
\item For player 3, we sum (\ref{Eq:Rc1}),(\ref{Eq:Rc2}),(\ref{Eq:Rc4}) and twice (\ref{Eq:Rc3}).
\item For player 4, we sum (\ref{Eq:Rc2}),(\ref{Eq:Rc3}), twice (\ref{Eq:Rc1}) and twice (\ref{Eq:Rc4}).
\item For player 5, we sum (\ref{Eq:Rc1}),(\ref{Eq:Rc3}),(\ref{Eq:Rc4}) and twice (\ref{Eq:Rc2}).
\item For player 6, we sum (\ref{Eq:Rc2}),(\ref{Eq:Rc4}), twice (\ref{Eq:Rc1}) and twice (\ref{Eq:Rc3}).
\item For player 7, we sum (\ref{Eq:Rc1}),(\ref{Eq:Rc2}),(\ref{Eq:Rc3}) and twice (\ref{Eq:Rc4}).
\end{itemize}

The probability that all players guess correctly under this strategy
is $3^{-4}$, i.e.\ the probability that (\ref{Eq:Rc1}),
(\ref{Eq:Rc2}), (\ref{Eq:Rc3}), (\ref{Eq:Rc4}) all hold. (It is not difficult to check that the conditions are linearly independent.)
Consequently
\[\gn(R_c) \geq |V(R_c)|+\log_3 \mb{P}[\Win(R_c,3,\mc{F})] = 9\]
as desired.
\end{proof}

For completeness we give the asymptotic guessing number of $R$ and
note that it does not match the fractional clique cover bound of
$20/3$ as claimed.

\begin{theorem}
We have $\gn(R) = 27/4$.
\end{theorem}

\begin{proof}
The Shannon bound of $R$ is $27/4$ (data files can be provided upon
request).

To show $\gn(R)\geq 27/4$ we will show $\gn(R,81)\geq 27/4$. By
Lemma \ref{Lem:BlowupEqual} this can be achieved if we can construct
a strategy on the guessing game $(R(4),3)$ which has a probability
of winning $3^{-13}$ (which implies $\gn(R(4),3)\geq 27$). To describe
such a strategy let us label the vertices of $R(4)$ such that the
four vertices that are constructed from blowing up $v\in V(R)$ are
labelled $v_a$, $v_b$, $v_c$, and $v_d$. Under this labelling our
strategy for $R(4)$ is to have the cliques $\{1_a,2_a,3_a\}$,
$\{1_b,4_a,5_a\}$, $\{1_c,6_a,7_a\}$, $\{2_b,3_b,9_a\}$,
$\{2_c,3_c,9_b\}$, $\{4_b,5_b,10_a\}$, $\{4_c,5_c,10_b\}$,
$\{6_b,7_b,8_a\}$ and $\{6_c,7_c,8_b\}$ play the complete graph
strategy, and the remaining $13$ vertices, which form a copy of $R_c$,
%$1_d$, $2_d$, $3_d$, $4_d$, $5_d$, $6_d$, $7_d$, $8_c$, $8_d$, $9_c$, $9_d$, $10_c$, $10_d$
to play the strategy for $R_c$ as described in the proof of Theorem.
\ref{Thm:Rc}.
\end{proof}

Now that we have shown that Conjectures \ref{Conj:LowerBoundSharp} and
\ref{Conj:UpperBoundSharp} are not true, we turn our attention to
other open questions. Due to the limited tools and methods currently
available, there are many seemingly trivial problems on guessing
games which still remain unsolved. One such problem is the
following.

\begin{problem}\label{Prob:SingleEdge}
Does there exist an undirected graph whose asymptotic guessing
number increases when a single directed edge is added?
\end{problem}

Adding a directed edge gives one of the players more information,
which cannot lower the probability that the players win. However,
surprisingly it seems extremely difficult to make use of the extra
directed edge to increase the asymptotic guessing number. An
exhaustive (but not completely rigorous) search on undirected graphs
with $9$ vertices or less did not yield any examples.

As such, we significantly weaken the requirements in Problem
\ref{Prob:SingleEdge} by introducing the concept of a Superman
vertex. We define a Superman vertex to be one that all other
vertices can see. I.e., given a digraph $G$, we call vertex $u\in V(G)$ a
\emph{Superman vertex} if $uv\in E(G)$ for all $v\in V(G)\setminus \{u\}$. We
can similarly define a Luthor vertex as one which sees all
other vertices. To be precise $u$ is a \emph{Luthor vertex} if $vu\in E(G)$
for all $v\in V(G)\setminus \{u\}$.

\begin{problem}\label{Prob:Superman}
Does there exist an undirected graph whose asymptotic guessing
number increases when directed edges are added to change one of the
vertices into a Superman vertex (or a Luthor vertex)?
\end{problem}

To change one of the vertices into a Superman or Luthor vertex will
often involve adding multiple directed edges, meaning the players
will have a lot more information at their disposal when making their
guesses. We again searched all undirected graphs on $9$ vertices or
less and remarkably still could not find any examples.

With the discovery of the graph $R$ and in particular the graph
$R_c$ we can show the answer is yes to Problem \ref{Prob:SingleEdge}
and consequently Problem \ref{Prob:Superman}. We define the
undirected graph $R_c^-$ to be the same as the graph $R_c$ but with
the undirected edge between vertices $3$ and $8$ removed. We also
define the directed graph $R_c^+$ to be the same as $R_c^-$ but with
the addition of a single directed edge going from vertex $3$ to vertex $8$.

\begin{theorem}
We have $\gn(R_c^-) = 53/6$ and $\gn(R_c^+) = 9$.
\end{theorem}

\begin{proof}
The Shannon bounds for $R_c^-$ and $R_c^+$ are $53/6$ and $9$
respectively (data files can be provided upon request).

We will prove $\gn(R_c^+)\geq 9$ by observing that the strategy for
$(R_c,3)$ (see the proof of Theorem \ref{Thm:Rc}) is a valid
strategy for $(R_c^+,3)$. With the exception of player $3$ all
players in $(R_c^+,3)$ have access to the same information they did
in $(R_c,3)$. Player $3$ however, now no longer has access to $a_8$. By studying the strategy player $3$ uses in $(R_c,3)$ we will see
that this is of no consequence. Summing conditions (\ref{Eq:Rc1}),
(\ref{Eq:Rc2}), (\ref{Eq:Rc4}), and twice (\ref{Eq:Rc3}), gives
\[
a_1+2a_2+a_3+2a_4+2a_{8'}+2a_{9} \equiv 0 \bmod 3,
\]
hence player $3$ guesses $-a_1-2a_2-2a_4-2a_{8'}-2a_{9} \bmod 3$ in
$(R_c,3)$. Since player $3$ makes no use of $a_8$ this validates our
claims.

We complete our proof by showing $\gn(R_c^-)\geq 53/6$. We know
$\gn(R_c^-)\geq \gn(R_c^-,3^6) = \gn(R_c^-(6),3)/6$ so it is enough
to show $\gn(R_c^-(6),3)\geq 53$. Since $R_c^-(6)$ had $78$ vertices we can do this by finding a strategy on
$(R_c^-(6),3)$ that wins with a probability of $3^{-25}$. To this end, let us
label the vertices of $R_c^-(6)$ such that the six vertices that are
constructed from blowing up $v\in V(R_c^-)$ are labelled $v_a$,
$v_b$, $v_c$, $v_d$, $v_e$, and $v_f$. Under this labelling, our
strategy for $R_c^-(6)$ is to play the complete graph strategy on
the cliques
\begin{align*}
&\{1_a,2_a,3_a\}, &&\{1_b,2_b,7_a\}, &&\{1_c,3_b,4_a\},
&&\{2_c,3_c,9'_a\}, &&\{4_b,5_a,10'_a\},\\
&\{4_c,5_b,10'_b\}, &&\{5_c,6_a,9'_b\}, &&\{6_b,7_b,8_a\},
&&\{6_c,7_c,8_b\}, &&\{8_c,9'_c,10'_c\},\\
&\{8_d,9'_d,10'_d\}, &&\{8_e,9'_e,10'_e\}, &&\{8_f,9'_f,10'_f\},
\end{align*}
and to play the $R_c$ strategy on the vertices
\begin{align*}
&\{1_d,2_d,3_d,4_d,5_d,6_d,7_d,8'_a,8'_b,9_a,9_b,10_a,10_b\},\\
&\{1_e,2_e,3_e,4_e,5_e,6_e,7_e,8'_c,8'_d,9_c,9_d,10_c,10_d\},\\
&\{1_f,2_f,3_f,4_f,5_f,6_f,7_f,8'_e,8'_f,9_e,9_f,10_e,10_f\}.
\end{align*}
The probability of winning in each of these 13 cliques is $3^{-1}$ while the probability of winning in each of the three copies of $R_c$ is $3^{-4}$. So the overall probability of winning is indeed $3^{-25}$, therefore completing the proof.
\end{proof}

We finish this section by considering a problem motivated by the
reversibility of networks in network coding. Given a digraph $G$,
let $\Reverse(G)$ be the digraph formed from $G$ by reversing all
the edges, i.e.\ $uv\in E(G)$ if and only if $vu\in E(\Reverse(G))$.

\begin{problem}\label{Prob:Reverse}
Does there exist a digraph $G$, such that $\gn(G)\neq
\gn(\Reverse(G))$.
\end{problem}

We were not able to solve this problem. We did however find a graph
$R^S$ for which the Shannon bound of $R^S$ and the Shannon bound of
$\Reverse(R^S)$ did not match. $R^S$ is simply the digraph formed by
making vertex $1$ in $R$ a Superman vertex. In other words, we add
three directed edges to $R$: the edge going from $1$ to $8$, from
$1$ to $9$, and from $1$ to $10$. Consequently $\Reverse(R^S)$ is
the graph formed by making vertex $1$ in $R$ a Luthor vertex. As
such, we will refer to it as $R^L$.

\begin{theorem}\label{Thm:RSandRL}
We have $\Sh(R^S) = 27/4 = 6.75$. For $R^L$ we have the following bounds:
\begin{enumerate}
\item $\Sh(R^L) = 34/5 = 6.8$.
\item $\ZY(R^L) = 61/9 = 6.777\ldots$
\item $\DFZ(R^L) = 359/53 = 6.773\ldots$
\item $\Ingl(R^L) = 27/4 = 6.75.$
\end{enumerate}
\end{theorem}

\begin{proof}
The proofs are given in data files which can be made available upon
request.
\end{proof}

From the strategy on $R$ we know that $\gn(R^S)\geq 27/4$ and
$\gn(R^L)\geq 27/4$. Hence we have $\gn(R^S)=27/4$. We do not
however know the precise value of $\gn(R^L)$ so it is possible that the
asymptotic guessing numbers of $R^S$ and $R^L$ do not match.

\section{Speeding up the computer search}\label{Sec:Speed}

In this section we mention a few of the simple tricks we used in order to
speed up the computer search which allowed us to search through all
the $10$ vertex graphs and find the graph $R$. We hope that this may
be of use to others continuing this research.

The majority of time spent during the searches was spent determining
the Shannon bound by solving a large linear program. By reducing the
number of constraints that we add to the linear program we can speed up
the optimisation. Given a graph on $n$ vertices a naive formation of
the linear program would result in considering all $2^{3n}$ Shannon
inequalities of the form
\[
\mbox{$H(A,C)+H(B,C)-H(A,B,C)-H(C)\geq 0$ for $A,B,C\subset X_G$}.
\]
However most of these do not need to be added to the linear program. In fact
it is sufficient to just include the inequalities given
by the following lemma.

\begin{lemma}\label{Lem:FewerInequalities}
Given a set of discrete random variables $X_G$, the set of Shannon
inequalities
\[
\mbox{$H(A,C)+H(B,C)-H(A,B,C)-H(C)\geq 0$ for $A,B,C\subset X_G$},
\]
is equivalent to the set of inequalities given by
\begin{enumerate}[(i)]
\item\label{Item:IneqPoset} $H(Y)\leq H(X_G)$ for $Y\subset X_G$ with $|Y|=|X_G|-1$.
\item\label{Item:IneqSubmodular} $H(Y)+H(Z)-H(Y\cup Z)-H(Y\cap Z)\geq 0$ for $Y,Z\subset X_G$ with $|Y|=|Z|=|Y\cap Z|+1$.
\end{enumerate}
\end{lemma}
Observe that for a graph on $n$ vertices there are $n$ inequalities of type (\ref{Item:IneqPoset}) and $n(n-1)2^{n-3}$ inequalities of type (\ref{Item:IneqSubmodular}). (Counting the inequalities of type (\ref{Item:IneqSubmodular}) is equivalent to counting the number of squares in the hypercube poset formed from looking at the subsets of $X_G$.)
%Observe that by choosing the set $Y$ first and then the set $Z$ we have
%\[ \sum_{k=0}^n (n-k)\binom{n}{k}\]
%inequalities of type (i). Choosing $Y \cap Z$ first and then extending to $Y$ and $Z$ we see that there are
%\[ \sum_{k=0}^n \binom{n}{k}\binom{n-k}{2}\]
%inequalities of type (ii). So in total we have
%\[ \frac{1}{2}\sum_{k=0}^n (n-k)(n-k+1)\binom{n}{k} = \frac{1}{2}\sum_{r=0}^n r(r+1)\binom{n}{r} \]
%inequalities to consider. This is equal to Since
%\[ \sum_{r=0}^n \binom{n}{r}x^{r+1} = x(1+x)^n\]
%by differentiating twice we have that
%\[ \sum_{r=1}^n r(r+1)\binom{n}{r}x^{r-1} = 2n(1+x)^{n-1} + n(n-1)x(1+x)^{n-2}\]
%and so there is a total of
%\[ \frac{1}{2}((2n)2^{n-1} + n(n-1)2^{n-2}) = n(n+3)2^{n-3}\]
%inequalities to consider.
Overall, this is about the cube root of the initial number of inequalities.

\begin{proof}[Proof of Lemma \ref{Lem:FewerInequalities}]
Setting $A=X_G$, $B=X_G$, and $C=Y$, shows that the Shannon inequalities
imply the set of inequalities described by (\ref{Item:IneqPoset}).
Setting $A=Y$, $B=Z$, and $C=Y\cap Z$, shows that the Shannon
inequalities imply (\ref{Item:IneqSubmodular}). 

To show (\ref{Item:IneqPoset}) and (\ref{Item:IneqSubmodular}) imply
the Shannon inequalities we will first generalise
(\ref{Item:IneqPoset}) and (\ref{Item:IneqSubmodular}).

We will begin by showing that (\ref{Item:IneqSubmodular}) implies
\[H(Y)+H(Z)-H(Y\cup Z)-H(Y\cap Z)\geq 0\]
for any $Y,Z\subset X_G$.
Let $Y\setminus (Y\cap Z) = \{Y_1,Y_2,\ldots,Y_n\}$ and $Z\setminus
(Y\cap Z) = \{Z_1,Z_2,\ldots,Z_m\}$, where $Y_1,\ldots,Y_n$ and
$Z_1,\ldots,Z_m$ are single discrete random variables. Define $Y'_i$
to be $\{Y_1,\ldots,Y_i\}$ for $1\leq i \leq n$ and $Y'_0 =
\emptyset$. We define $Z'_i$ similarly. Finally let $X_{i,j} =
(Y\cap Z)\cup Y'_i\cup Z'_j$, and note that $X_{0,0} = Y\cap Z$,
$X_{n,0} = Y$, $X_{0,m}=Z$, and $X_{n,m} = Y\cup Z$. By
(\ref{Item:IneqSubmodular}) we have
\[
0\leq
\sum_{i=0}^{n-1}\sum_{j=0}^{m-1}[H(X_{i+1,j})+H(X_{i,j+1})-H(X_{i+1,j+1})-H(X_{i,j})].
\]
Here, the right hand side is telescopic and simplifies to the desired expression
\[H(Y)+H(Z)-H(Y\cup Z)-H(Y\cap Z).\]

Next we will generalise (\ref{Item:IneqPoset}) to show that for any $Y\subset Z \subset X_G$ with $|Y| = |Z|-1$ we have $H(Y) \leq H(Z)$. Let us define $\overline{Z}$ to be $X_G\setminus Z$. Then, by the generalised version of (\ref{Item:IneqSubmodular}) we know that
\[
H(Z) + H(Y \cup \overline{Z}) - H(Z \cup (Y\cup\overline{Z})) - H(Z \cap (Y\cup\overline{Z})) \geq 0
\]
which simplifies to
\begin{align}\label{Eqn:GeneralisedI}
H(Z) + H(Y \cup \overline{Z}) - H(X_G) - H(Y) \geq 0.
\end{align}
Observe that $|Y\cup\overline{Z}| = |X_G|-1$, so (\ref{Item:IneqPoset}) tells us that $H(X_G)-H(Y\cup\overline{Z})\geq 0$ which when added to (\ref{Eqn:GeneralisedI}) gives the inequality $H(Z)-H(Y)\geq 0$ as required.

We can now further generalise (\ref{Item:IneqPoset}) to show that for any
$Y\subset Z\subset X_G$ we have that $H(Y)\leq H(Z)$. To do this, let
$Z\setminus Y = \{Z_1,Z_2,\ldots,Z_n\}$, where $Z_1,\ldots,Z_n$ are
single discrete random variables. Then, by repeated applications of our generalisation of (\ref{Item:IneqPoset}) we have
\[
H(Y)\leq H(Y,Z_1)\leq H(Y,Z_1,Z_2)\leq \cdots \leq
H(Y,Z_1,\ldots,Z_n) = H(Z).
\]

It is now a trivial matter to show that (\ref{Item:IneqPoset}) and
(\ref{Item:IneqSubmodular}) imply Shannon's inequality. Simply set
$Y=A\cup C$ and $Z=B\cup C$ in the generalised version of
(\ref{Item:IneqSubmodular}) to get
\[
H(A,C)+H(B,C)-H(A,B,C)-H(A\cap B, C)\geq 0
\]
and since $H(A\cap B, C)\geq H(C)$ by the improved version of
(\ref{Item:IneqPoset}), the result follows.
\end{proof}

It is also worth mentioning that $H(\emptyset)=0$ together with the
Shannon inequalities imply $H(Y,Z)\leq H(Y)+H(Z)$ for disjoint
$Y,Z$. Hence, the constraints $H(X)\leq |X|$ for all $X$ in the
Shannon bound linear program are not all necessary and can be
reduced to $H(X)\leq |X|$ for $|X|=0$, or $1$.

When determining each graph's asymptotic guessing number, the natural
approach is to calculate the lower bound using the fractional clique
cover number, then calculate the Shannon bound and check if they
match. However the linear program that gives us the fractional
clique cover number also gives us a regular fractional clique cover
from which an explicit strategy can be constructed. It is easy to
convert this strategy into a feasible point of the Shannon bound
linear program. Hence we can save a significant amount of time by
simply checking if this feasible point is optimal, rather than by
calculating the Shannon bound from scratch. Note that we check for
optimality by solving the same Shannon bound linear program with the
modification that we remove those constraints for which equality
is not achieved by the feasible point.

The modified Shannon bound linear program is still the most time
consuming process in the search, so ideally we would like to avoid
it when possible. Christofides and Markstr\"om \cite{Christofides&Markstrom11} show that for an
undirected graph $G$
\begin{align*}
\gn(G) \leq |V(G)|-\alpha(G),
\end{align*}
where $\alpha(G)$ is the number of vertices in the maximum
independent set. This can be interpreted as a simple consequence of
the fact that removing players increases the probability the
remaining players will win. (If the probability of winning
decreased, the players could just create fictitious replacement
players before the game started.) As such we present a simple
generalization of this result.

\begin{lemma}\label{Lem:RemoveVertices}
Given a digraph $G$ and an induced subgraph $G'$,
\[
\max_{\mc{F}}\mb{P}[\Win(G,s,\mc{F})]\leq
\max_{\mc{F}}\mb{P}[\Win(G',s,\mc{F})]
\]
or equivalently $\gn(G,s)\leq |V(G)|-|V(G')|+\gn(G',s)$. Hence
\[
\gn(G) \leq |V(G)|-|V(G')|+\gn(G').
\]
\end{lemma}

We do not provide a proof as it is trivial. Note that the result
$\gn(G)\leq |V(G)|-\alpha(G)$ is a simple corollary of this result
as an independent set has a guessing number of $0$.

Given a graph $G$, if we can find a subgraph such that the upper
bound given in Lemma \ref{Lem:RemoveVertices} matches the fractional
clique cover bound, then we have determined the asymptotic guessing
number, and can avoid an expensive Shannon bound calculation. This
approach is particularly fast when doing an exhaustive search as all
the smaller graphs will have had their asymptotic guessing numbers
already determined.

One issue with this method is that if we are looking for a
counterexample to the Shannon bound being sharp, there is a
possibility that we may miss them because we avoided calculating the
Shannon bound for every graph. Consequently to alleviate our fear we
need the following result.

\begin{lemma}\label{Lem:SubgraphShannonBound}
Given a digraph $G$ and an induced subgraph $G'$, we have
\[ \Sh(G) \leq |V(G)|-|V(G')| + \Sh(G'). \]
\end{lemma}

\begin{proof}
It is sufficient to prove the result only for induced subgraphs $G'$
for which $|V(G)|-|V(G')|=1$, as the result then follows by induction on $|V(G)|-|V(G')|$. Let $u\in V(G)$ be the vertex that is
removed from $G$ to produce $G'$.

The Shannon bound for $G'$ comes from solving a linear program, and
as such the solution to the dual program naturally gives us a proof
that $H(X_{G'}) \leq \Sh(G')$. In
particular, this proof consists of summing appropriate linear
combinations of the constraints. Suppose that in each such contraint we replace $H(X)$ with $H(X,X_u)-H(X_u)$ for every $X\subset X_{G'}$. This effectively
would replace constraints from the linear program for $G'$ with
inequalities which are implied from the linear program for $G$. For
example, $H(X)\geq 0$ for $G'$, would become $H(X,X_u)-H(X_u)\geq 0$
for $G$ (which is true by Shannon's inequality). As another example, $H(X)\leq |X|$
becomes $H(X,X_u)-H(X_u)\leq |X|$ (which is true as $H(X,X_u)\leq
H(X)+H(X_u)\leq |X|+H(X_u)$). This shows that all constraints in Theorem \ref{Thm:ShannonEntropy} of types (1) and (2) can be replaces as claimed. The same happens for constraints of types (3) and (4). Consequently, under this transformation, the proof that $H(X_{G'}) \leq \Sh(G')$ 
becomes a proof that $H(X_{G})-H(X_u) \leq \Sh(G')$. Since $H(X_u)\leq 1$ the result immediately
follows.
\end{proof}

We have seen that by removing vertices from a graph $G$ we make the
game easier allowing us to upper bound $\gn(G)$. Another way we can
make the game easier is by adding extra edges to $G$. Consequently
we can avoid the Shannon bound calculation by also using the
asymptotic guessing number of supergraphs of $G$ which have the same
number of vertices as $G$.

We end this section by considering the problem of how to calculate
the non-Shannon bounds, i.e.\ the Zhang-Yeung bound, the
Dougherty-Freiling-Zeger bound, and the Ingleton bound. They all
involve inequalities on $4$ variables and consequently a naive
approach is to add at least $2^{4n}$ inequalities to the linear
program, where $n$ is the order of the graph. Unfortunately such a
linear program is far too large to be computationally feasible. Our
approach is given by the following algorithm:
\begin{enumerate}
\item Let $\mc{C}$ be the set of Shannon bound constraints.
\item\label{Item:Prog4Var} Solve the linear program which consists only of constraints
$\mc{C}$.
\item\label{Item:Check4Var} Check if the solution satisfies all required $4$ variable information
inequalities (e.g.\ the Zhang-Yeung inequalities if we are
calculating the Zhang-Yeung bound).
\begin{enumerate}
\item If all the inequalities are satisfied then terminate,
returning the objective value.
\item If one of the inequalities is not satisfied add this
constraint to $\mc{C}$ and go back to \ref{Item:Prog4Var}.
\end{enumerate}
\end{enumerate}

We note that due to the large number of inequalities, Step
\ref{Item:Check4Var} can take a while. So it is advisable to add 
some extra constraints to the linear program to limit the search 
to a solution which is symmetric under the automorphisms of the
graph (there always exists such a solution due to the linearity of
the problem). This extra symmetry can be used to avoid checking a
significant proportion of the inequalities in Step
\ref{Item:Check4Var}.

\section{Open Problems}\label{Sec:Problems}
%There are many open questions that arise from our research. We only mention a few of them:

%Even though the fractional clique cover bound is not always the same as the asymptotic guessing number, it is interesting to know whether they are equal for specific families of graphs. One such inter is that of the triangle-free graphs. If $G$ is a triangle-free graph then Lemma \ref{Lem:BoundOnFractionalClique} tells us that $\kappa_f(G) \geq |V(G)|/2$. If Conjecture \ref{Conj:LowerBoundSharp} were true, then we would have $\gn(G) \leq |V(G)|/2$. 

%\begin{conjecture}
%If $G$ is a triangle-free undirected graph then $\gn(G) \leq |V(G)|/2$.
%\end{conjecture}

Problem \ref{Prob:Reverse} asks whether there exists an irreversible guessing game, i.e.\ a guessing game $G$ such that $\gn(G)\neq\gn(\Reverse(G))$. This can be answered in the affirmative if $\gn(R^L)$ can be shown to be strictly larger than $\frac{27}{4}=6.75$. Unfortunately, this might be hard to prove as it would establish the existence of a non-linear guessing strategy that improves the lower bound we derived. 

It would also be interesting to determine the exact value of $\gn(R^-)$ as $R^-$ according to our calculations is the only undirected graph on at most $10$ vertices whose guessing number remains undetermined. Any lower bound that implies $\gn(R^-) > \frac{20}{3}$ would show that there exists a non-linear guessing strategy that outperforms the fractional clique cover strategy for $R^-$.

\begin{comment}
Both $R^L$ and $R^-$ may provide the first examples of undirected graphs for which there exists no optimal strategy which is linear, which would also be of great interest.
\end{comment}

\section{Acknowledgements}
We would like to thank Peter Cameron and Peter Keevash. An extended abstract of this paper appeared in \cite{ConferencePaper} and we would like to thank the three anonymous referees for their useful comments. This work was partly supported by EPSRC ref: EP/H016015/1.

\end{document}